\documentclass[aps,amsfonts,twoside,amssymb,superscriptaddress,twocolumn]{revtex4}

\usepackage[latin1]{inputenc}
\usepackage[T1]{fontenc}
\usepackage{amsmath,amssymb,amstext,amsfonts,amsxtra}
\usepackage{amsthm}
\usepackage{graphicx}
\usepackage{mathrsfs}
\usepackage{enumerate}
\usepackage{bbm}
\usepackage{hyperref}
\usepackage{color}

\newcommand{\tr}{\operatorname{tr}}

\def\idty{{\leavevmode\rm 1\mkern -5.4mu I}} 
\def\id{{\rm id}}

\def\ket #1{\vert #1\rangle}

\def\ketbra #1#2{\vert #1\rangle \langle #2\vert}

\def\tr{\mathop{\rm tr}\nolimits}

\def\EE{{\mathcal E}}
\def\cX{{\mathcal X}}

\def\PP{{\mathbb P}}
\def\QQ{{\mathbb Q}}

\def\cZ{{\mathcal Z}}

\def\PP{{\mathbb P}}

\newcommand\prob[2]{\ensuremath{\mathrm{Pr}_{#2}[#1]}}
\newcommand\pr[1]{\ensuremath{\mathrm{Pr}[#1]}}

\newcommand\Hmax[2]{\ensuremath{H_{\max} \left(#1 \vert #2 \right)}}

\newcommand\SHmin[2]{\ensuremath{H^{\epsilon}_{\min} \left(#1 \vert #2 \right)}}
\newcommand\SHmax[2]{\ensuremath{H^{\epsilon}_{\max} \left(#1 \vert #2 \right)}}

\newcommand*{\cP}{\mathcal{P}}

\newtheorem{prop}{Proposition}

\newtheorem{de}{Definition}

\newcommand{\pass}{\textnormal{pass}}

\def\cH{{\mathcal H}}
\def\cS{{\mathcal S}}

\newcommand{\pab}{p_{\textnormal{abort}}}

\newcommand{\leakEC}{\textnormal{leak}_{\textnormal{EC}}}
\newcommand{\ppass}{p_{\textnormal{pass}}}

\begin{document}

\title{Continuous Variable Quantum Key Distribution: Finite-Key Analysis of Composable Security against Coherent Attacks}

\author{F. Furrer} \email{fabian.furrer@itp.uni-hannover.de} \address{Institut f\"ur Theoretische Physik, Leibniz Universit\"at Hannover Appelstra\ss e 2, 30167 Hannover, Germany }
\author{T. Franz}  \address{Institut f\"ur Theoretische Physik, Leibniz Universit\"at Hannover Appelstra\ss e 2, 30167 Hannover, Germany }
\author{M. Berta} \address{Institut f\"ur Theoretische Physik, ETH Z\"urich, 8093 Z\"urich, Switzerland}
\author{A. Leverrier} \address{Institut f\"ur Theoretische Physik, ETH Z\"urich, 8093 Z\"urich, Switzerland}
\author{V.~B. Scholz }  \address{Institut f\"ur Theoretische Physik, Leibniz Universit\"at Hannover Appelstra\ss e 2, 30167 Hannover, Germany }
\author{M. Tomamichel} \address{Institut f\"ur Theoretische Physik, ETH Z\"urich, 8093 Z\"urich, Switzerland}
\author{R.~F. Werner} \address{Institut f\"ur Theoretische Physik, Leibniz Universit\"at Hannover Appelstra\ss e 2, 30167 Hannover, Germany }

\begin{abstract}
We provide a security analysis for continuous variable quantum key distribution
protocols based on the transmission of two-mode squeezed vacuum states measured via
homodyne detection. We employ a version of the entropic uncertainty relation for
smooth entropies to give a lower bound on the number of secret bits which can be
extracted from a finite number of runs of the protocol. This bound is valid
under general coherent attacks, and gives rise to keys which are composably
secure. For comparison, we also give a lower bound valid under the assumption of
collective attacks. For both scenarios, we find positive key rates using
experimental parameters reachable today.
\end{abstract}


\maketitle

Quantum key distribution (QKD) is one of the first ideas from quantum information theory for turning quantum paradoxes into applications, see~\cite{Scarani_review_2009} and references therein. The task in QKD is to generate a shared key, secret from any eavesdropper (Eve), between two distant parties (Alice and Bob) using communication over a public quantum channel and an authenticated classical channel. Many different implementations of QKD have been proposed, each one with individual strengths and weaknesses. Early proposals are based on exchanging qubits, and are part of the family of discrete variable (DV) QKD protocols. Continuous variable (CV) protocols have later been proposed and offer the possibility to use standard telecom technologies (see~\cite{Weedbrook2011} and references therein), in particular, they do not require photon counters.

A generic QKD protocol starts with the distribution of, say, $N$ quantum states between the honest parties which are then measured according to the rules of the protocol. A certain part of the measurement outcomes is then used to estimate Eve's information about the remaining data from which a key of length $\ell$ is generated by classical post-processing. The goal of a finite-key security analysis is to prove that the key is secure against any wiretapping strategy of Eve, up to a small failure probability. This is in contrast to the study of asymptotic rates in which perfect security in the limit for $N$ to infinity is considered.

Eve's knowledge can be bounded by the probability that she correctly guesses Alice's measurement outcomes. This is expressed by the conditional smooth min-entropy~\cite{Renner_Phd} of the data from which the key is generated given Eve's quantum system. This ensures composable security~\cite{canetti01}, i.e., the protocol can securely be combined with other composeably secure cryptographic protocols. Since the actual state is not known, the smooth min-entropy has to be bounded for the worst case compatible with the observed measurement data. This is in general a hard task and often simplified by additional assumptions about the power of the eavesdropper.
Instead of allowing the most general, \emph{coherent} attack on the quantum communication between Alice and Bob, the eavesdropper is often restricted to \emph{collective} attacks, meaning that every signal is attacked with the same quantum operation. Under this assumption, Alice and Bob can employ state tomography to bound Eve's information and to ensure security. In the case of DV QKD, these security proofs can then often be lifted to security proofs against coherent attacks using the exponential de Finetti theorems~\cite{renato_nature} or the post-selection technique~\cite{Renner_Postselection}.

Most security analysis for CV protocols neglect finite-key effects and consider asymptotic rates by using the Devetak-Winter formula~\cite{Devetak08012005}(see \cite{Berta11} for an infinite dimensional version). We are only aware of~\cite{Leverrier2010}, where a first finite-key analysis for specific protocols under the assumption of collective Gaussian attacks was provided. Security against coherent attacks was considered in~\cite{GottesmanPreskill,Cerf05} based on entanglement purification protocols, but without a quantitive analysis. The transfer of the exponential de Finetti technique to the infinite-dimensional setting is very subtle. This is because exponential de Finetti theorems do in general not hold in infinite-dimensional systems~\cite{christandl07}, but only under additional assumptions~\cite{Renner_Cirac_09}. It is often argued that, using these results, much of the DV theory can be transferred to CV systems. Unfortunately, this approach provides only pessimistic finite-key rate estimates.

Recently, a more direct approach to prove DV QKD secure against coherent attacks was presented in~\cite{Lim11}, which is based on an entropic uncertainty relation with quantum side information for smooth entropies~\cite{Tomamichel11}. This uncertainty relation gives a bound on Eve's information about Alice's measurement outcomes in terms of the correlation between Alice and Bob. The relation between security in QKD and uncertainty relations has also been employed in~\cite{grosshans04,koashi06}. Based on the recent extension of the smooth entropy formalism to the infinite-dimensional setting~\cite{Furrer10,Berta11}, it is the objective of this letter to apply the above reasoning to an entanglement based CV protocol using two-mode squeezed vacuum states measured via homodyne detection.

\emph{Security Definition and Key Rates.}---\,A generic QKD protocol between two honest parties, Alice (A) and Bob (B) either aborts or outputs a key which consists of strings $S_A$ and $S_B$ on Alice's and Bob's side, respectively. We denote by $E$ the information which is wiretapped during the run of the protocol by an attack on the quantum channel. For CV systems this is modeled on an infinite-dimensional Hilbert space. The state of $S_A$ and $E$ can be described as a classical quantum state
\begin{align}
  \omega_{S_AE}=\sum_{s}\ketbra ss \otimes \omega_{E}^s\ ,
\end{align}
where $\omega_{E}^s$ are states on Eve's system. Three requirements have to be fulfilled by an ideal protocol: correctness, secrecy and robustness. Correctness is achieved when the output on Alice's and Bob's side agree, $S_A = S_B$. Secrecy of a key means that $S_A$ is uniformly distributed and independent of $E$ and thus given by $\omega^{\mathrm{id}}_{S_AE} = \tau_{S_A} \otimes \sigma_E$, with $\tau_{S_A}$ the uniform mixture of keys, and $\sigma_E$ an arbitrary state on the $E$ system. A protocol is called secure if it is both correct and secret. Finally, we call an ideal protocol robust if it never aborts when Eve is passive.

In reality, we can only hope to achieve an almost ideal protocol. For small parameters $\epsilon_c$, $\epsilon_s$ and an abortion probability $\pab$, we require that the protocol is $\epsilon_c$-correct, i.e.~ $\pr{S_A \neq S_B} \leq \epsilon_c$, and that the protocol is $\epsilon_s$-secret, i.e.~ $(1-\pab)\, \inf_{\sigma} \frac{1}{2} \Vert \omega_{S_AE}-\tau_{S_A} \otimes \sigma_E \Vert \leq \epsilon_s$. Note that a protocol which always aborts is secure. Thus we may impose an additional requirement on the robustness, e.g., $\pab < 1$. This security definition also ensures that the protocol is secure in the framework of composable security~\cite{canetti01}, in which different cryptographic protocols can be combined without compromising the overall security. We note that this is not the case for security definitions that require only a small mutual information between the eavesdropper and the key~\cite{RennerKoenig05}.

The measurement step of a QKD protocol produces a pair of raw keys, $X_A$ and $X_B$, held by Alice and Bob. If the protocol does not abort, the secret keys $S_A$ and $S_B$ are extracted using classical error correction and privacy amplification schemes. We do not discuss the error correction scheme here and simply assume that it will leak $\leakEC$ bits of information about the key to the eavesdropper. The correctness is checked using a hash function evaluated on both resulting strings which leads to an additional leakage of order $O(\log \frac{1}{\epsilon_c})$ \cite{Lim11}.

In the privacy amplification step, two-universal hash functions are used to compress the raw key to a final length of $\ell$ bits. Roughly speaking, this reduces Eve's knowledge about Alice's key by $N-\ell$ bits. Hence, choosing $\ell$ sufficiently small ensures that Eve has no information about the resulting bit strings and the key is independent of E. Formally, Eve's uncertainty (or lack of knowledge) is measured in terms of the probability that she can guess Alice's raw key $X_A$, i.e.\ the conditional min-entropy $H_{\min}(X_A|E)$ (see Appendix \ref{App1:def_entropies} for a formal definition).  In particular, the resulting key is $\epsilon_s$-secret if~\cite{Renner_Phd,Berta11,Tomamichel10}
\begin{align}\label{GeneralKeyLength}
  \ell \lesssim H_{\min}^{\epsilon}(X_A|E)_{\omega} - \leakEC - O(\log \frac{1}{\epsilon_s\epsilon_c})\ \, ,
\end{align}
where $\epsilon \propto \epsilon_s/\pab$. Here, the smooth min-entropy, $H_{\min}^{\epsilon}(X_A|E)$, is the optimization of the min-entropy over states which are
$\epsilon$ close to $\omega_{X_AE}$, where $\omega_{X_AE}$ denotes the joint state prior to the classical post-processing conditioned on the event that the protocol does not abort. We derive lower bounds on this entropy for the following protocol.

\emph{The Protocol.}---\,
The analysis of coherent and collective attacks can widely be treated in parallel. We consider a trusted source located in Alice's lab that produces an entangled state by mixing two squeezed vacuum states on a balanced beam splitter. We assume that each beam consists of only one bosonic mode. Alice sends one beam to Bob whereupon both perform a homodyne measurement. They choose uniformly at random between two canonically conjugated quadrature observables, amplitude and phase, such that Alice's and Bob's outcomes are maximally correlated whenever their choice agree. In the case of collective attacks they additionally perform measurements to estimate the covariance matrix. We further assume that the states generated by the source have tensor product form and that the probability that Alice measures an amplitude or phase quadrature is larger than $\alpha$ ($\hbar =1$) is bounded by $p_\alpha$. This is possible since the source is trusted and located in Alice's lab.

After all measurements are performed, the two parties reveal their measurement choices. In the case of coherent attacks, they discard the data in which they have measured different quadratures ending up with a string of $N$ measurement results. Then, they divide the continuous outcome range of the quadrature measurements into intervals $ (-\infty, -\alpha\!+\!\delta], (-\alpha\!+\!\delta,-\alpha+2\delta], \ldots, (\alpha\!-\!\delta, \infty)$ where we assume for simplicity that $2\alpha/\delta\in\mathbb N$. We denote the outcome alphabet by $\cX=\{1,2,...,2\alpha/\delta\}$. A random sample $X_A^{pe}, X_B^{pe} \in \cX^{k}$ of length $k$ are used for parameter estimation, in which they check the quality of their correlation by computing the average distance $d(X_A^{pe},X_B^{pe}) = \frac{1}{k}\sum_{i=1}^{k}\vert X_{A,i}^{pe}-X_{B,i}^{pe}\vert$ where $X_A^{pe}=(X_{A,i}^{pe})_{i=1}^k$ and$X_B^{pe}=(X_{B,i}^{pe})_{i=1}^k$. If $d(X_A^{pe},X_B^{pe})$ is smaller than $d_0$ they proceed and otherwise they abort the protocol. In case the test is passed, they use the remaining data $X_A,X_B\in\cX^n$ ($n=N-k$) as the raw key and execute the error correction and privacy amplification protocol as discussed in the paragraph before. For collective attacks, the strings $X_A\in\cX^n$ and $X_B\in\cX^n$ are generated as for coherent attacks but the remaining data (before the binning) is used to estimate the covariance matrix. This also includes the one in which Alice and Bob measured different quadratures.

\emph{Analysis for Coherent Attacks.}---\,
The goal is to bound the smooth min-entropy conditioned on the event that the protocol does not abort. For that we use an infinite-dimensional version of the entropic uncertainty relation for smooth entropies with side information~\cite{Berta11}, combining the uncertainty principle for complementary measurements with monogamy of entanglement. It states that Eve's information about the measurement outcomes $X_A$ can be bounded by using the the complementary of the measurements and the correlation between $X_A$ and $X_B$. In particular, if Alice and Bob are highly correlated after measuring e.g.,~the phase quadrature, then Eve's knowledge about the outcome of the amplitude measurement is nearly zero, since the observables are maximally complementary. We measure this correlation strength by the smooth max-entropy $H_{\max}^{\epsilon}(X_A|X_B)$, which characterizes the amount of information Alice has to send Bob to retrieve $X_A$. This leads to the bound (see Appendix \ref{App2:UncertaintyRelation})
\begin{align}\label{uncertainty}
  H_{\min}^{\epsilon}(X_A|E)_{\omega} \geq  n\log \frac{1}{c(\delta)} - H_{\max}^{\epsilon'}(X_A|X_B)_{\omega}\ ,
\end{align}
where $c(\delta)$ is the overlap of the two conjugated quadrature measurements on an interval of length $\delta$ which is well approximated by $c(\delta)\approx \delta^2/(2\pi )$ for small $\delta$. Equation~(\ref{uncertainty}) assumes a uniformly random choice of measurement settings. Since projectors onto intervals $(-\infty,-\alpha]$ and $[\alpha,\infty)$ would lead to a trivial state-independent uncertainty relation, the probability of this event has to be estimated using $p_\alpha$. In equation \ref{uncertainty} this is included in the change of the smoothing parameter from $\epsilon$ to $\epsilon'$.

This reduces the problem to upper bounding the smooth max-entropy between $X_A$ and $X_B$, which can be done by $n\cdot\log\gamma(d(X_A,X_B))$, where $\gamma$ is a function arising from a large deviation consideration (see Appendix \ref{App3:Statistics}). Using sampling theory, the quantity $d(X_A,X_B)$ can then, with high probability, be estimated by $d(X^{pe}_A,X^{pe}_B)$ plus a correction $\mu$, which quantifies its statistical deviation to $d(X_A,X_B)$ and depends on $p_\alpha$, $k$ and $n$. Since the protocol aborts if $d(X^{pe}_A,X^{pe}_B)>d_0$, we obtain the following formula for the key length:
\noindent \emph{For parameters $k,p_\alpha,\delta,d_0$, an $\epsilon_s$-secret key of length}
\begin{equation*}\label{KeyCoherent}
\ell = n [\log \frac{1}{c(\delta)}-\log \gamma(d_0 + \mu)]- \leakEC - O(\log \frac{1}{\epsilon_s\epsilon_c})\ .
\end{equation*}
\emph{can be extracted.}

\begin{figure}[h]\begin{center}\includegraphics*[width=8.8cm]{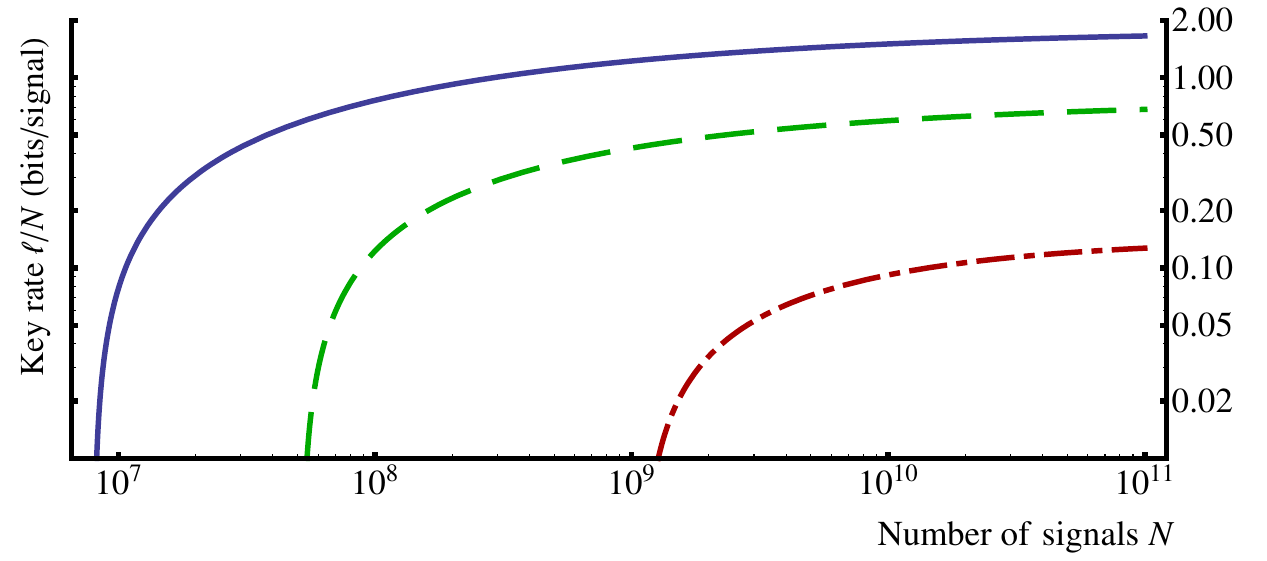}\caption{\label{fig:coherent} Key rate $\ell/N$ against coherent attacks for an input squeezing/antisqueezing of $11$dB/$16$dB and additional symmetric losses of $0\%$ (solid line), $10\%$ (dashed line) and $20\%$ (dash-dotted line). For the chosen security parameters see the main text.}\end{center}\end{figure}

We assume that the source in Alice's lab is trusted and that her measurement device is described by projections onto two canonical variables. Note that the measurement device on Bob's side need not to be trusted, except that measurements on different signals commute. Hence, the additional reference signal (local oscillator) used by Bob for homodyne detection is covered by our security analysis.
Placing the trusted source in Alice's lab also implies that the analysis is not compatible with reverse reconciliation.

We calculate the correlation between $X_A$ and $X_B$ under the assumption of an identically and independently distributed source producing states with an input squeezing of $11$dB and antisqueezing of $16$dB. Squeezing at this level has been realized in an experiment at 1550nm~\cite{Eberle11}. Our noise model consists of loss and excess noise, where the latter is set to be $1\%$ as it is mainly due to the classical data acquisition (see Appendix \ref{App4:ErrorModel}). The leakage term is estimated assuming an error correction efficiency of $0.95$~\cite{Leverrier_log_dis}. In Fig.~\ref{fig:coherent} the resulting key rates $\ell/N$ are plotted for different symmetric losses. We have set security parameters $\epsilon_s=\epsilon_c=10^{-6}$. The optimization over the other free parameters is done numerically for each $N$. Typical values for $N=10^9$ are $k=10^8$, $\alpha=52$ and $\delta=0.01$.

\emph{Analysis for Collective Attacks}---\,
Under the assumption of collective attacks, the state between Alice, Bob, and Eve has tensor product structure, $\omega_{ABE}^{\otimes N}$, enabling statistical estimations of the covariance matrix of $\omega_{AB}$. However, we do not cover the statistical details here and simply introduce confidence sets $\mathcal{C}_{\epsilon_{pe}}$, which ensure that whenever the protocol does not abort the covariance matrix $\Gamma_{AB}$ of $\omega_{AB}$ lies in $\mathcal{C}_{\epsilon_{pe}}$ with probability at least $1-\epsilon_{pe}$.
Hence, we have to give a lower bound on the smooth min-entropy $H_{\min}^{\epsilon}(X_A|E)_{\omega^{\otimes n}}$ over all states with a covariance matrix $\Gamma_{AB}\in\mathcal{C}_{\epsilon_{pe}}$. The smooth min-entropy is evaluated on the classical quantum state $\omega_{X_AE}$ which is obtained from $\omega_{AB}$ by taking a purification $\omega_{ABE}$ and applying the discretized quadrature measurement on the $A$ system.

\begin{figure}[h]\begin{center}\includegraphics*[width=8.8cm]{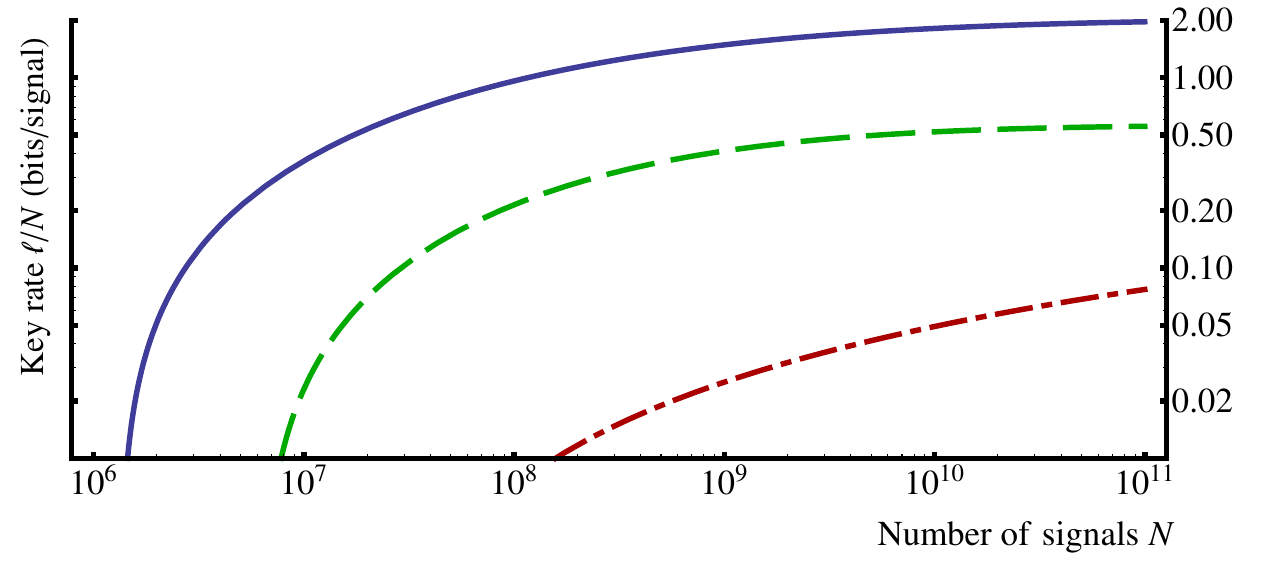}\caption{\label{fig:collective} Key rate $\ell/N$ against collective Gaussian attacks for losses of $0\%$ (solid line), $15\%$ (dashed line), $25\%$ (dash-dotted line). Squeezing strength and security parameters are chosen as in the case of coherent attacks.}\end{center}\end{figure}

We employ the quantum equipartition property of the smooth min-entropy~\cite{Tomamichel08} for infinite-dimensional systems~\cite{Furrer10}, stating that for large $n$, $H_{\min}^{\epsilon}(X_A|E)_{\omega^{\otimes n}}$ approaches the conditional von Neumann entropy $H(X_A|E)_{\omega}$. More precisely, we have
\begin{equation}\label{eq:AEP}
H_{\min}^{\epsilon}(X_A|E)_{\omega^{\otimes n}} \geq n\cdot H(X_A|E)_{\omega} - \sqrt{n}\cdot\Delta\ ,
\end{equation}
where $\Delta$ is a function of $\epsilon$, $\delta$ and $\alpha$ (see Appendix \ref{App5:AEP}). Using that the minimum of $H(X_A|E)_{\omega}$ over all states with a fixed covariance matrix $\Gamma_{AB}$ is attained for the corresponding Gaussian state $\omega^{\Gamma_{AB}}$ (see Appendix \ref{App6:Extremality} and \cite{Cerf2006}), we get the following formula for the key length.

\noindent
\emph{For parameters $k,\alpha,\delta$, an $(\epsilon_s\!+\!\epsilon_{pe})$-secret key of length}
\begin{align*}\label{KeyCollective}
 n\cdot\!\inf_{\Gamma\in\mathcal{C}_{\epsilon_{pe}}} \!\! H(X_A|E)_{\omega^{\Gamma}}-\sqrt{n}\cdot\Delta - \leakEC - O(\log \frac{1}{\epsilon_s\epsilon_c})
\end{align*}
\emph{can be extracted assuming collective attacks.}

\begin{figure}[h]\begin{center}\includegraphics*[width=8.8cm]{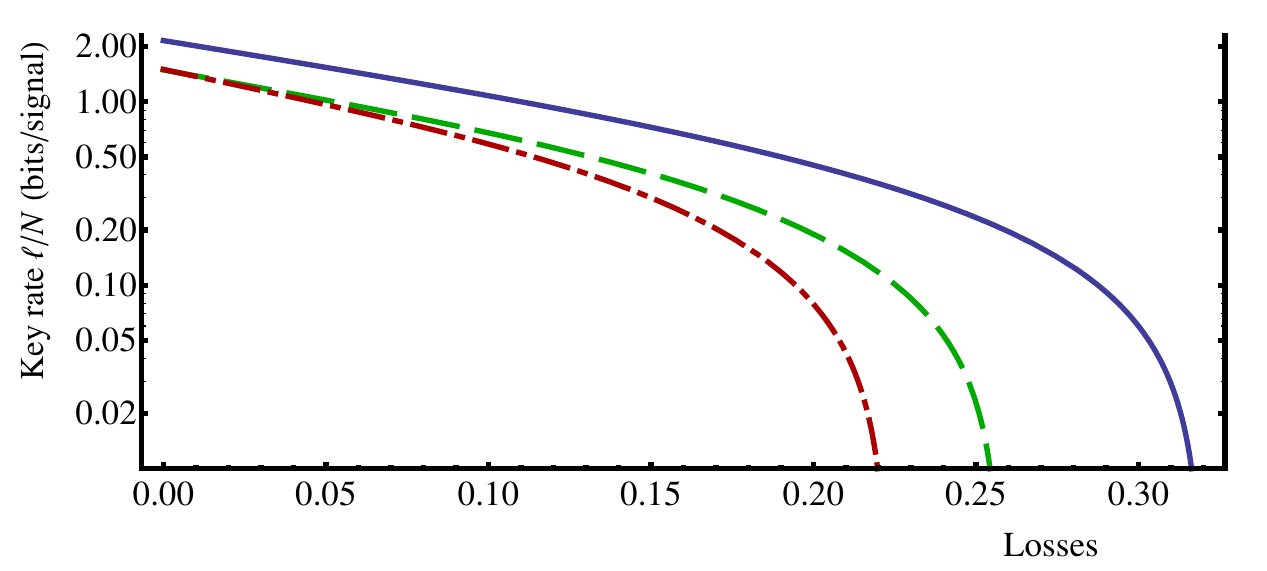}\caption{\label{fig:KeyRateVsLosses} Key rate versus losses secure against coherent attacks at $N=10^9$ (dot-dashed line), collective Gaussian attacks at $N=10^9$ (dashed line), and the Devetak-Winter rate~\cite{Devetak08012005} for perfect information reconciliation (solid line). Squeezing strength and security parameters are chosen as in the case of coherent attacks.}\end{center}\end{figure}

To evaluate this finite-key bound numerically, we need explicit expressions for the confidence sets. For this, we use results from~\cite{Leverrier2010}, which assume collective Gaussian attacks. We computed the key rates $\ell/N$ in Fig.~\ref{fig:collective} for the same squeezing strength and loss model as in the case of coherent attacks. The detailed calculation of $H(X_A|E)_{\omega^{\Gamma}} $ can be found in Appendix \ref{App7:Entropy}. For simplicity, we assumed a constant binning of $\delta$ over the entire outcome range ($\alpha= \infty$). In contrast to the case of coherent attacks, reverse reconciliation is possible and can increase the key rate essentially if asymmetric losses are assumed (which we do not discuss here). In Fig.~\ref{fig:KeyRateVsLosses}, we plotted the key rate for coherent and collective Gaussian attacks in dependence of the losses, and compare them with the Devetak-Winter rate~\cite{Devetak08012005,Berta11} for perfect error correction.

\emph{Discussion and Outlook.}---\,We provided a finite-key security analysis for a continuous variable QKD protocol and obtain a composable secure  positive key rate against coherent attacks for experimentally feasible parameters. 
The comparison with the finite-key rate against collective attacks shows that the gap is relatively small compared to the finite-size effects. This is due to the fact that the uncertainty relation is almost tight for the two-mode squeezed states. The reason that the key rates allow for only small amounts of losses is because of the direct reconciliation in the error correction protocol. Hence, an extension of the proof technique against coherent attacks to a reverse reconciliation error correction protocol would be desirable.
In order to relax the assumptions in the security proof against coherent attacks, it would be interesting to study the overlap for more realistic models of the quadrature measurements, which may include a continuum of modes. Moreover, our arguments might also be applicable to other CV QKD schemes \cite{Grosshans02,Weedbrook04}.

\emph{Acknowledgments.}---\,We thank R.~Renner for suggesting this work, and R.~Garc\'ia-Patr\'on and I. Cirac for helpful discussions. F.F acknowledges support from the LUH GRK 1463. T.F, V.B.S, and R.F.W acknowledge support from the DFG (grant WE-1240/12-1), BMBF project QuOReP, EU project Q-ESSENCE, and the research cluster QUEST. M.B is supported by the SNF (grant PP00P2-128455), and the DFG (grants CH 843/1-1 and CH 843/2-1). A.L., M.B, V.B.S and M.T are supported by the SNF through the National Centre of Competence in Research `Quantum Science and Technology'.

\begin{appendix}
\section{Smooth Min- and Max-Entropies}\label{App1:def_entropies}

For the sake of completeness, we give here a formal definition of the smooth conditional min- and max-entropies and present the basic properties used in the following. For a detailed discussion consult e.g.~\cite{koenig-2008, Tomamichel08, Tomamichel09} for the finite-dimensional case and \cite{Furrer10,Berta11} for the infinite-dimensional case. In the following, $\cH$ always denotes a separable Hilbert space and $\cS(\cH)$ the state space associated to $\cH$, which consists of all positive semi-definite trace class operators on $\cH$ with trace $1$. Furthermore, we define $\cS_{\leq}(\cH)$ to be the set of all non-normalized states, that is, positive semi-definite trace class operators with trace smaller or equal to $1$. We indicate different subsystems by labels and denote a state on $\cH_{AB}=\cH_A\otimes\cH_B$ by $\omega_{AB}$ and its reduced state on $\cH_A$ simply by $\omega_A$. Classical systems are denoted by $X,Y,Z$ and are described by embedding the classical degrees of freedom into a Hilbert space w.r.t.~a fixed orthonormal basis. This allows to read the following definitions for quantum as well as classical systems.

\begin{de}\label{Def:MinEntorpy}
For $\omega_{AB}\in\cS_{\leq}(\cH_{AB})$, we define the min-entropy of $A$ conditioned on $B$ as
\begin{equation*}
H_{\min}(A|B)_\omega = \sup_{\sigma_B\in \cS(B)} \sup\{ \lambda \in \mathbbm{R}|\, \omega_{AB} \leq 2^{-\lambda} \idty_A \otimes \sigma_B  \}\ .
\end{equation*}
\end{de}
The min-entropy of a classical-quantum state $\omega_{XB}$ characterizes the optimal guessing probability of the classical variable $X$ given the quantum system $B$~\cite{koenig-2008}.

The purified distance between two states $\omega,\rho\in\cS_{\leq}(\cH)$ is defined~\cite{Tomamichel09} as $\cP(\omega,\rho)=\sqrt{1-F(\omega,\rho)}$ where $F(\omega,\rho)=\big(\tr(\vert\sqrt{\omega}\sqrt{\rho}\vert)+\sqrt{(1-\mathrm{tr}[\sigma])(1-\mathrm{tr}[\rho])}\big)^2$ denotes the generalized fidelity.

\begin{de}\label{Def:SmoothEntorpies}
For $\omega_{AB}\in\cS_{\leq}(\cH_{AB})$ and $\epsilon \geq 0$, we define the $\epsilon$-smooth min-entropy of $A$ conditioned on $B$ as
\begin{equation}
H^\epsilon_{\min}(A|B)_\omega = \sup H_{\min}(A|B)_{\tilde\omega}\ .
\end{equation}
where the supremum is taken over all $\tilde\omega_{AB}\in\cS_{\leq}(\cH_{AB})$ with $\cP(\omega_{AB},\tilde\omega_{AB})\leq \epsilon$.
The $\epsilon$-smooth max-entropy of $A$ conditioned on $B$ is defined as
\begin{equation}\label{def:smoothMax}
H^\epsilon_{\max}(A|B)_\omega = -H^\epsilon_{\min}(A|C)_\omega\ .
\end{equation}
where $\omega_{ABC}$ is an arbitrary purification of $\omega_{AB}$.
\end{de}
One can show that the definition of the smooth max-entropy is independent of the choice of the purification. As for the min-entropy, we denote the non-smoothed version ($\epsilon=0$) of the max-entropy simply by $H_{\max}(A|B)_\omega$. The smooth max-entropy can also be expressed as the optimization of the max-entropy over $\epsilon$-close states, that is,
\begin{equation}\label{Eq:SmoothMax}
H^\epsilon_{\max}(A|B)_\omega = \inf H_{\max}(A|B)_{\tilde\omega}\ ,
\end{equation}
where the infimum is taken over all $\tilde\omega_{AB}\in\cS_{\leq}(\cH_{AB})$ with $\cP(\omega_{AB},\tilde\omega_{AB})\leq \epsilon$.
These entropies satisfy the data processing inequality saying that whenever the system B is manipulated with a quantum operation $\EE:\cS(\cH_B)\rightarrow \cS(\cH_C)$, the entropy can only increase
\begin{eqnarray}\label{eq:DataProcessingIneq}
\SHmin {A}{B}_{\omega} &\leq& \SHmin {A}{C}_{\id_A\otimes\EE(\omega)} \\
\SHmax {A}{B}_{\omega} &\leq& \SHmax {A}{C}_{\id_A\otimes\EE(\omega)}\ .\label{eq:DataProcessingIneqHmax}
\end{eqnarray}


\section{Derivation of the Uncertainty Relation}\label{App2:UncertaintyRelation}

Let us assume that the protocol parameters $\alpha$ and $\delta$ are fixed. For simplicity, we further assume that $M:=2\alpha/\delta$ is in $\mathbb N$. In the protocol Alice and Bob both measure the projectors of the quadrature measurements on the intervals $I_{1}=(-\infty,-\alpha +\delta]$, $I_{2}=(-\alpha+\delta,-\alpha + 2\delta]$, ..., $I_{M}=(\alpha-\delta,\infty)$. Let us denote the corresponding outcome alphabet by $\cX=\{1,2,,...,M\}$, which by definition is of size $|\cX|=2\alpha/\delta$. Let us introduce another partition of $\mathbb R$ into intervals $\{\tilde I_k\}_{k\in\mathbb N}$ of equal length $\delta$ such that $\tilde I_{k}=I_k$ for $k\in\cX\backslash \{1,M\}$. In the following we denote the projection onto the interval $I$ of the  spectrum of the phase and amplitude operator of Alice by $Q_A(I)$ and $P_A(I)$.

We can assume that they first distribute all the subsystems on which they perform the measurements. Let us denote the state shared between Alice, Bob and Eve on which they produce the sifted $N$ measurements by $\omega_{A^NB^NE}$, where $N$ denotes the number of subsystems. In the parameter estimation step they check that the average distance of the random sample of $k$ measurements $X^{pe}_A,X^{pe}_B\in\cX^k$ satisfies
\begin{equation}
d(X^{pe}_A,X^{pe}_B) \leq  d_0.
\end{equation}
Note that this test can be written as a projector $\Pi^{\pass}_k$ which only acts non-trivially on the $k$ subsystems used in the parameter step. If this condition holds, they pursue with the protocol otherwise they abort. Let us denote by $\omega_{A^nB^nE}$ the quantum state on the remaining $n$ subsystems conditioned on the event that the parameter estimation test passes.

Alice chooses now for each subsystem uniform at random between phase and amplitude measurements. This can be modeled by introducing a random variable $Z^n=(Z_1,...,Z_n)$ independently and identically distributed according to the uniform distribution, where $Z_i$ takes values $0$ or $1$ depending on whether Alice measures phase or amplitude in the ith run. Let us denote the uniform distribution over $\cZ^n=\{0,1\}^n$ by $u$ and by $\{ \ket {z^n}\}_{z^n\in\cZ^n}$ an orthonormal basis of a Hilbert space. The random measurement choice of Alice can now be modeled by introducing the state
\begin{equation}
\omega_{Z^nA^nB^nE}=\sum_{z^n\in\cZ^n}u(z^n) \ketbra {z^n}{z^n}\otimes \omega_{A^nB^nE} \; ,
\end{equation}
and the positive operator valued measure (POVM) $\{\Pi_A^{l^n}(z^n)\otimes \ketbra {z^n}{z^n}\}_{z^n\in\cZ^n,l^n\in\mathbb \cX^n}$ acting on $A^n$ and $Z^n$, where
\begin{equation}
\Pi_A^{l^n}(z^n)=\bigotimes_i\Pi_A^{l^n_i}(z^n_i) \; ,
\end{equation}
with $\Pi_A^{i}(0)=Q_A(I_{i})$ and $\Pi_A^{i}(1)=P_A(I_{i})$ for $i\in\cX$. Hence, $z^n_i$ determines whether phase or quadrature is measured. Let us denote the post-measurement state obtained by measuring the state $\omega_{A^nB^nEZ^n}$ by the POVM $\{\Pi_{k^n}(z^n)\otimes \ketbra {z^n}{z^n}\}$ by $\omega^n_{X_AB^nEZ^n}$. Here, $X_A$ takes values in $\cX^k$ and denotes the random variable which describes the distribution of the keys. Note that all parties are assumed to hold a copy of the variable $Z$ since the measurement choices have been revealed in the sifting phase.
Additionally, we introduce a similar POVM for the projections onto the spectrum of Alice's phase and amplitude measurements onto the intervals $\{\tilde I_i\}_{i\in\mathbb N}$ and denote them by
\begin{equation}
\tilde\Pi_A^{l^n}(z^n)=\bigotimes_i\tilde\Pi_A^{l_i}(z_i) \; ,
\end{equation}
where $\tilde\Pi_A^{i}(0)=Q_A(\tilde I_{i})$ and $\tilde \Pi_A^{i}(1)=P_A(\tilde I_{i})$ for $i\in\mathbb N$. The corresponding post-measurement state is denoted by $\tilde\omega^n_{ X_AB^nEZ^n}$. Note that here the distribution over $X_A$ can take values in $\mathbb N^n$.

The main idea in the security proof is to apply an uncertainty relation with quantum side information for the smooth min- and max-entropy~\cite{Tomamichel11}. For that, it is important that the measurement are maximally complementary. In the case of the POVM $\{\Pi_{k^n}(z^n)\otimes \ketbra {z^n}{z^n}\}$ this is a problem since the projectors of the phase and amplitude measurements onto the big intervals $I_1$ for $i=1,M$ almost commute. The idea is now that by trusting the source, we can estimate the (purified) distance between the states $\omega^n_{X_ABEZ^n}$ and $\tilde\omega^n_{X_ABEZ^n}$. For the state $\tilde\omega^n_{X_ABEZ^n}$, we can then obtain a non-trivial uncertainty relation since all projectors have only support on an interval of length $\delta$. In particular, it follows that
\begin{equation}\label{eq:uncertainty}
\SHmin {X_A}{EZ^n}_{\tilde\omega} \geq -n\log c - \SHmax {X_A}{B^nZ^n}_{\tilde\omega}
\end{equation}
where $c=\sup_{i,j} \Vert \sqrt{ Q_A(\tilde I_{i})}\sqrt{P_A(\tilde I_{j})}\Vert^2$. The inequality in this form is proven for the finite-dimensional case in~\cite[Corollary 7.6]{TomamichelPhD}. The generalization to infinite dimensions is straightforward by using the techniques from~\cite{Berta11}. It turns out the the overlap $c$ only depends on the length of the intervals and is given by
\begin{equation*}
c(\delta)= \frac{\delta^{2}}{2\pi}\cdot S_0^{(1)}(1,\frac{\delta^2}{4})^2 \, ,
\end{equation*}
where  $S_n^{(1)}(\cdot,u)$ denotes the radial prolate spheroidal wave function of the first kind (see~\cite{Kiukas10} and references therein).

In a next step, we estimate the distance between $\omega^n_{X_ABE\cZ}$ and $\tilde\omega^n_{X_ABE\cZ}$. According to the main text, we assume that the source produces a state which is independently and identically for each run. That is, the state has tensor product form $\omega_{A^N}=\omega_{A}^{\otimes N}$. Furthermore, if we set $\bar p_\alpha=1-p_\alpha$ we have by assumption that the source satisfies
\begin{equation}
 \tr \Big[ \omega_{ A} Q_A([-\alpha,\alpha])\Big] \geq \bar p_{\alpha},
\end{equation}
as well as $ \tr\big[\omega_{ A} P_A([-\alpha,\alpha])\big] \geq \bar p_{\alpha}$. Let us now define $\Lambda=\mathbb N \backslash \cX$ and for every $z^n\in [0,1]^n$ the projector
\begin{equation}
\Pi_A^\Lambda(z^n)=\sum_{l^n\in \Lambda}\tilde\Pi_A^{l^n}(z^n) \;
\end{equation}
which corresponds to the event where at least one of the quadrature measurements exceeds $\alpha$.  Since
\begin{equation*}
\omega_{A^n}= \frac{1}{\ppass}\tr_{A^kB^N}\Big(\Pi^{\pass}_k \omega_{A^NB^N}\Big)\leq \frac{1}{\ppass} \omega^{\otimes n}_A \;
\end{equation*}
with $p_{\mathrm{pass}}=1-p_{\mathrm{abort}}$, we obtain for every $z^n\in \cZ^n$
\begin{equation}
 \tr \Big[ \omega_{ A^n}\Pi_A^\Lambda(z^n)\Big] \leq \frac{1-\bar p_{\alpha}^n}{\ppass}.
\end{equation}
We can now bound the fidelity for a fixed $z^n\in\cZ^n$ by
\begin{align*}
F(\omega^{z^n}_{X_AB^nE},\tilde\omega^{z^n}_{X_AB^nE}) &\geq (1- \tr \Big[ \omega_{ A^n}\Pi_A^\Lambda(z^n)\Big] )^2\\
& \geq 1- 2\tr \Big[ \omega_{ A^n}\Pi_A^\Lambda(z^n)\Big]  \\
& \geq 1 - 2\frac{1-\bar p_{\alpha}^n}{\ppass} \; ,
\end{align*}
where $\omega^{z^n}_{X_AB^nE}$ and $\tilde\omega^{z^n}_{X_AB^nE}$ denote the normalized states conditioned on the event $z^n$.
Since now the fidelity between $\omega^n_{X_AB^nEZ^n}$ and $\tilde \omega^n_{X_AB^nEZ^n}$ is just the average over $z^n\in\cZ^n$, we obtain by the definition of the purified distance (see Section~\ref{App1:def_entropies})
\begin{equation}\label{eq:DistStates}
\cP(\omega^n_{X_AB^nEZ^n},\tilde \omega^n_{X_AB^nEZ^n}) \leq {\frac{f(n,p_\alpha)}{\sqrt{\ppass}}} \; ,
\end{equation}
where $f(p_\alpha,n)=\sqrt{2(1-(1-p_\alpha)^n)}$.

The bound in~(\ref{eq:DistStates}) can now be used to bound the smooth min- and max-entropy by
\begin{align*}\label{eq:MinEntropyEpsilon}
H_{\min}^{\epsilon+ \tilde\epsilon}(X_A|EZ^n)_{{\omega}} &\geq H_{\min}^{\epsilon}(X_A|EZ^n)_{\tilde \omega} \\
	-H_{\max}^{\epsilon+\tilde\epsilon'}(X_A|B^nZ^n)_{\tilde\omega} &\geq - H_{\max}^{\epsilon}(X_A|B^nZ^n)_{{\omega}} \; ,
\end{align*}
where $\tilde\epsilon=f(p_\alpha,n)/\sqrt{\ppass}$. We simply used the Definition~\ref{Def:SmoothEntorpies} and the fact that the purified distance can only decrease by tracing out a subsystem. In combination with the uncertainty relation in~(\ref{eq:uncertainty}), we arrive at
\begin{equation*}\label{eq:uncertaintyQKD}
H_{\min}^{\epsilon+ 2\tilde\epsilon}(X_A|EZ^n)_{{\omega}} \geq -n\log c(\delta) - H_{\max}^{\epsilon}(X_A|B^nZ^n)_{{\omega}} \;.
\end{equation*}
Applying the data processing inequality to the max-entropy $ H_{\max}^{\epsilon}(X_A|B^nZ^n)_\omega \leq H_{\max}^{\epsilon}(X_A|X_B)_\omega$, we obtain the final uncertainty relation used in Equation~(3) in the main text. Note that we assumed in the main text that $Z^n$ is included in $E$.


\section{Statistical Analysis for Coherent Attacks}\label{App3:Statistics}

The goal is to show that if the protocol does not abort and thus, satisfies $d(X^{pe}_A,X^{pe}_B) \leq  d_0$, the smooth max-entropy in Equation~(3) can be bounded by
\begin{equation}\label{BoundHmax}
 H_{\max}^{\epsilon'}(X_A|X_B)_{\omega}\leq n \log \gamma(d_0 + \mu_0)\ ,
\end{equation}
where
\begin{equation*}
\gamma(t)= (t+\sqrt{1+t^2})\Big(\frac{t}{\sqrt{1+t^2}-1}\Big)^{t}\ ,
\end{equation*}
and
\begin{equation}
\mu_0=|\cX|\sqrt{\frac{N(k+1)}{nk^2}\log\frac{1}{\epsilon_s/4 - 2f(p_\alpha,n)}}\; .
\end{equation}
Note that the alphabet $\cX$ satisfies $|\cX|=\lceil2\frac{\alpha}{\delta}\rceil$ and $\epsilon'={\epsilon_s}/({4\ppass}) - 2f(p_\alpha,n)/\sqrt{\ppass}$~\footnote{The value of $\epsilon$ in Equation~(4) can be chosen as ${\epsilon_s}/({4\ppass})$ and that the other term $2f(p_\alpha,n)/\sqrt{\ppass}$ comes from the uncertainty relation derived in Sectoin~\ref{App2:UncertaintyRelation}.}. The proof is divided into two steps and follows closely the arguments in~\cite{Lim11}. First we derive a bound on the smooth max-entropy, and then we estimate the probability that $d(X_{A},X_B)\geq d(X_A^{pe},X_B^{pe})+\mu$.

\begin{prop}\label{prop:MaxEntropyBound}
Let $\cX$ be a finite alphabet, $\PP(x,x')$ a probability distribution on $\cX^n \times\cX^n$ for some $n\in\mathbb{N}$, $d_0>0$ and $\epsilon>0$. If $\PP$ satisfies $\prob{d(x,x')\geq d_0}{\PP} \leq \epsilon^2$, then
\begin{equation*}
\SHmax X{X'}_{\PP} \leq n \log \gamma(d_0)\ ,
\end{equation*}
where
\begin{equation*}
\gamma(t) = (t+\sqrt{1+t^2})\Big(t/[\sqrt{1+t^2}-1]\Big)^{t}\ .
\end{equation*}
\end{prop}

\begin{proof}
We first note that the smooth max-entropy is obtained by taking the infimum over non-smooth max-entropies of all states which are $\epsilon$-close in purified distance~(\ref{Eq:SmoothMax}). Let us define the probability distribution
\begin{equation*}
\QQ(x,x') = \begin{cases} \frac{\PP(x,x')}{\prob{d(x,x')\leq d_0}{\PP}}, & \mbox{if } d(x,x')\leq d_0
\\ 0 , & \mbox{else }  \end{cases}
\end{equation*}
and note that $F(\PP,\QQ)=\prob{d(x,x')\leq d_0}{\PP}$. Hence, it follows that $\cP(\PP,\QQ)=\sqrt{\prob{d(x,x')\geq d_0}{\PP}}\leq \epsilon$. Using that the $0$-R\'enyi-entropy is bigger than the max-entropy~\cite{Tomamichel10}, we obtain
\begin{equation*}
\SHmax X{X'}_{\PP} \leq \Hmax X{X'}_{\QQ} \leq H_0(X|X')_{\QQ}\ .
\end{equation*}
The conditional 0-R\'enyi entropy of the distribution $\QQ$ is then given by~\cite[Remark 3.1.4]{Renner_Phd}
\begin{align*}
H_0(X|X')_{\QQ} & = \max_{x'}\log |\{x\in\cX^n \ ; \; \QQ(x,x')\neq 0\}| \\
& \leq \log |\{x\in\mathbb Z ^n \  ; \; \sum_{i=1}^n|x_i|\leq nd_0 \}|\ .
\end{align*}
For any $\lambda>0$ we estimate
\begin{align*}
|\{x\in\mathbb Z ^n \  ; \; \sum_{i=1}^n|x_i|\leq nd_0 \}| & \leq \sum_{x\in\mathbb Z^n} \exp[\lambda (nd_0-\sum_{i=1}^n|x_i|)]  \\
&= e^{\lambda n d_0} \Big(\sum_{z\in\mathbb Z}e^{-\lambda |z|}\Big)^n\\
& = \Big(e^{\lambda d_0}\frac{1+e^{-\lambda}}{1-e^{-\lambda}}\Big)^n\ .
\end{align*}
By optimizing over $\lambda>0$, one finds that  $|\{x\in\mathbb Z ^n \  ; \; \sum_{i=1}^n|x_i|\leq nd_0 \}| \leq \gamma(d_0)^n$. This completes the proof.
\end{proof}

Now, we have to estimate the probability that $d(X_{A},X_B)  \geq d(X_A^{pe},X_B^{pe}) + \nu$ conditioned on the event that the protocol does not abort. Since the probability that the protocol passes is $p_{\mathrm{pass}}$, we find according to Bayes' theorem that
\begin{align*}
&\pr{d(X_{A},X_B) \geq d(X_A^{pe},X_B^{pe}) + \nu | ``\mathrm{pass}"}\\
&\leq \frac{1}{p_{\mathrm{pass}}}\pr{d(X_{A},X_B) \geq d(X_A^{pe},X_B^{pe}) + \nu}\ .
\end{align*}
Deriving a bound on $\pr{d((X_{A},X_B) \geq d(X_A^{pe},X_B^{pe}) + \nu}$ is a standard problem from random sampling without replacement. We have that $X_A^{pe},X_B^{pe}\in\cX^k$ is a random sample of all measurements $X_A^{tot},X_B^{tot}\in\cX^N$. The quantity of interest is $|x_A^i-x_B^i|$, where $x_A^i\in X_A^{tot}$ and $x_B^i\in X_B^{tot}$. For this we denote the population mean by $d_{tot}=d(X_A^{tot},X_B^{tot})$, the sample mean by $d_{pe}=d(X_A^{pe},X_B^{pe})$, and for the raw key $d_{key}=d(X_{A},X_B)$. Note that these are related via
\begin{equation}\label{eq:sample1}
Nd_{tot} = kd_{pe} + nd_{key}\ .
\end{equation}

We consider the runs of the protocol as a probabilistic process and treat $d_{tot}$ as a random variable. We first use the bound from~\cite{Serfling74} to obtain
\begin{equation*}
\pr{d_{key} \geq a + \nu|d_{tot}=a}\leq e^{-2n\nu^2 \frac{N}{|\cX|^2(k+1)}}\ ,
\end{equation*}
which is independent of $a$. Here, we used that the maximal value of $|x_A^i-x_B^i|$ is given by $|\cX|$. Using Eq.~\eqref{eq:sample1}, we can compute
\begin{eqnarray*}
\pr{d_{key}\geq d_{pe} + \nu} = \pr{d_{key}\geq d_{tot} + \frac{k}{N}\nu} \\
= \sum_{a} \pr{d_{tot}=a} \cdot \pr{d_{key} \geq a + \frac{k}{N} \nu|d_{tot}=a}  \\
\leq e^{-2\nu^2 \frac{nk^2 }{|\cX|^2N(k+1)}}\ .
\end{eqnarray*}
Hence, together with Proposition~\ref{prop:MaxEntropyBound} and the fact that the protocol aborts for $d(X_A^{pe},X_B^{pe})>d_0$, we arrive at
\begin{equation*}
\SHmax X{X'}_{\PP} \leq  n \log \gamma(d_0+\nu)
\end{equation*}
for
\begin{equation*}\label{eq:SamplePart}
\nu=|\cX|\sqrt{\frac{N(k+1)}{nk^2}\log\frac{1}{\epsilon\cdot\sqrt{p_{\mathrm{pass}}}}}\ .
\end{equation*}
In the protocol, we have to bound the smooth max-entropy for a smoothing parameter $\epsilon'=\epsilon=\frac{\epsilon_s}{4\ppass} - 2f(p_\alpha,n)/\sqrt{\ppass}$. Since $\ppass\leq 1$, we can bound $\nu\leq \mu$ and obtain the bound in Equation~(\ref{BoundHmax}).


\section{The Error Model}\label{App4:ErrorModel}

We consider a symmetric two parameter error model, using the loss $\mu_{\mbox{loss}}$ and excess noise $\mu_{\mbox{en}}$. The loss is our main source of noise and is equivalent to replacing a certain amount of signal by vacuum. The excess noise corresponds to a classical noise added by the data acquisition system and can in principle be made arbitrary small by using appropriate equipment. Both effects are gaussian noise sources and are expressed as action on the covariance matrix by $\Gamma \rightarrow (1-\mu_{\mbox{loss}}) \Gamma + (\mu_{\mbox{loss}} + \mu_{\mbox{en}})  \Gamma_{\mbox{vac}}$, where $\Gamma_{\mbox{vac}}$ denotes the covariance matrix of the vacuum state.


\section{Asymptotic Equipartition Property}\label{App5:AEP}

We use~\cite[Proposition 8]{Furrer10}, which states that for $\epsilon>0$, $n\geq\frac{8}{5}\log\frac{2}{\epsilon^{2}}$, and any quantum state $\omega_{AB}$ for which $H(A)_\omega$ is finite, we have
\begin{align*}
&H_{\min}^{\epsilon}(A|B)_{\omega^{\otimes n}}\geq n\cdot H(A|B)_{\omega}-\sqrt{n}\cdot\\
&4\log(2^{-\frac{1}{2}H_{\min}(A|B)_{\omega}}+2^{\frac{1}{2}H_{\max}(A|B)_{\omega}}+1)\sqrt{\log\frac{2}{\epsilon^{2}}}\ .
\end{align*}
In our case, we are interested in the classical quantum state $\omega_{X_AE}$ for which $H(X_A)_\omega$ is finite and the formula applies.
Let us now simplify the last term in the above inequality. Let $\omega_{X_AEC}$ be an arbitrary purification of $\omega_{X_AE}$, we have according to the definition of the max-entropy~(\ref{def:smoothMax})
\begin{equation*}
-H_{\min}(X_A|E)_{\omega}= H_{\max}(X_A|C)_{\omega} \leq H_{\max}(X_A)_{\omega}
 \end{equation*}
where the last inequality is due to the data processing inequality~(\ref{eq:DataProcessingIneqHmax}). Furthermore, we can also use the data processing inequality~(\ref{eq:DataProcessingIneqHmax}) to bound the max-entropy $H_{\max}(X_A|E)_{\omega} \leq H_{\max}(X_A)_{\omega}$. Using this two estimations, we obtain
\begin{align*}
&2^{-\frac{1}{2}H_{\min}(X_A|E)_{\omega}}+2^{\frac{1}{2}H_{\max}(X_A|E)_{\omega}} \leq 2^{\frac{1}{2}H_{\max}(X_A)_{\omega}+1}\ .
\end{align*}
Hence, we finally arrive at
$H_{\min}^{\epsilon}(X_A|E)_{\omega^{\otimes n}}\geq n\cdot H(X_A|E)_{\omega}-\sqrt{n}\cdot\Delta$
with
\begin{align}
\Delta=4\log(2^{\frac{1}{2}H_{\max}(X_A)_{\omega}+1}+1)\sqrt{\log\frac{2}{\epsilon^{2}}}\ ,
\end{align}
which is used in (4) of the main paper. Note that $\Delta$ only depends on $\epsilon$ and the measurement distribution on Alice's side. Since we assume in our setup a known source in Alice's lab this can be directly calculated.


\section{Gaussian Extremality}\label{App6:Extremality}

In the following we show that the infimum $\inf_{\omega} H(X_A|E)_{\omega}$ taken over all states $\omega_{AB}$ with covariance matrix $\Gamma$ is attained for the Gaussian representative. Even though the argument is in analogy to~\cite{Cerf2006}, we give it here for the sake of completeness. See also~\cite{Navascues06} for a similar result.

The main tool is the result from~\cite{Wolf06} which classifies functions which are optimized by Gaussian states. In particular, if one can show that a function $f(\omega_{AB})$ is (i) lower semi-continuous in trace norm, (ii) invariant under local unitary transformations, and (iii) strongly superadditive, i.e.~$f(\omega_{ABA'B'})\geq f(\omega_{AB})+f(\omega_{A'B'})$ where equality holds if $\omega_{ABA'B'}=\omega_{AB}\otimes\omega_{A'B'}$, then it follows that $f(\omega_{AB})\geq f(\omega_{AB}^{\Gamma})$. Here, $\omega_{AB}^{\Gamma}$ denotes the Gaussian representative of the family of states with same covariance matrix $\Gamma$.

Consider now the function $f(\omega_{AB})=H(X|E)_{\omega}$ where $\omega_{ABE}$ is a purification of $\omega_{AB}$ and $\omega_{XBE}$ is obtained by applying the measurement used in our protocol on the A system. The conditional von Neumann entropy is defined in accordance with~\cite{Kuznetsova}, that is, $H(A|B)_\rho=H(A)_\rho - H(\rho_{AB} ||\rho_A\otimes\rho_B)$ where $H(\rho||\sigma)$ denotes the relative entropy. In this definition we require that $H(A)_\rho$ is finite. Note that the classical alphabet $\cX$ on which $\omega_{X}$ is defined is finite such that $H(X)_\omega$ is always finite and the conditional entropy is well-defined. Because $0\leq H(X|B)_{\rho} \leq H(X)_{\rho}\leq \log |\cX|$ holds for any finite-dimensional B systems, we obtain the same result for infinite-dimensional Hilbert spaces via the finite-dimensional approximation property of the conditional von Neumann entropy as shown in~\cite{Kuznetsova}.

We show now that $f(\omega_{AB})=H(X|E)_{\omega}$ satisfies the properties (i)-(iii) from which the extremality of the Gaussian state follows. The properties (i) and (ii) are obtained in a similar way as in~\cite{Cerf2006}. In order to show property (iii) one takes a purification $\omega_{ABA'B'E}$ of $\omega_{ABA'B'}$, which is of course also a purification of $\omega_{AB}$ and $\omega_{A'B'}$. The following chain of inequalities for the von Neumann entropies
\begin{align*}
H(XX|E)_\omega  = & \; H(X|X'E)_\omega + H(X'|XE)_\omega \\
&   \, + I(X:X'|E)_\omega \\
\geq & \; H(X|A'B'E) + H(X|ABE)
\end{align*}
holds for finite-dimensional systems due to $ I(X:X'|E)_\omega\geq 0$ and since $X$ ($X'$) is obtained from $AB$ ($A'B'$) via a trace preserving completely positive map. But this can be lifted to infinite-dimensions via the finite-dimensional approximation property~\cite{Kuznetsova} as the entropies are all finite. Hence, we obtain the strong superadditivity
\begin{align*}
f(\omega_{ABA'B'}) &= H(XX|E)_\omega \\
&\geq H(X|A'B'E) + H(X|ABE) \\
&= f(\omega_{AB}) + f(\omega_{A'B'}) \, .
\end{align*}
The equality in the case of $\omega_{AB}\otimes\omega_{A'B'}$ follows from the additivity of the von Neumann entropy.


\section{Calculation of $H(X_A|E)$ for Discretized Measurements}\label{App7:Entropy}

In order to compute the bound on the key length secure against collective attacks as given in the main paper, we have to calculate $H(X_A|E)_{\omega}$ for a two mode squeezed Gaussian state $\omega_{AB}$. For the proper definition and the properties of conditional von Neumann entropies for infinite-dimensional systems, we refer to~\cite{Kuznetsova}. Let $\omega_{ABC}$ be a Gaussian purification of $\omega_{AB}$ with $\omega_E$ a two mode Gaussian state. We first rewrite the entropy as
\begin{align*}
H(X_A|E)_{\omega} & = H(X_AE)_{\omega}-H(E)_{\omega} \\
& = H(E|X_A)_{\omega}+H(X_A)_{\omega} - H(AB)_{\omega}\ ,
\end{align*}
where we used that $\omega_{ABE}$ is pure and therefore $H(E)=H(AB)$. Note also that in our case the alphabet $\cX$ is finite. Since $\omega_{AB}$ is a two mode Gaussian state the entropy $H(AB)_{\omega}$ is just a function of the symplectic invariants and can be calculated~\cite{serafini04}.

For the computation of the other entropies, we assume for simplicity that the correlations in amplitude and phase are symmetric, and do the calculation for the amplitude measurement with corresponding operator denoted by $X$. The measurement operators for a projection onto the interval $I_k$, $k\in\cX$, are described by $E_k=\mu_{x}(I_k)$, where $\mu_{x}$ is the spectral measure of $X$. The post-measurement states are then given by $\omega^k_{ABE}=1/p_k(E_k \omega_{ABE}E_k^{\dagger})$, where $p_k=\tr\left[\omega_{ABE}E_k\right]$. The entropy $H(X_A)_{\omega}$ is the Shannon entropy of the classical distribution $\{p_k\}$.

Let us turn to the estimation of $H(E|X_A)_{\omega}$. First, we note that
\begin{equation*}\label{CondEntr}
H(E|X_A)_{\omega}=\sum_{k}p_k H(E)_{\omega^k}\ ,
\end{equation*}
which reduces the problem to calculate $H(E)_{\omega^k}$ for every $k\in\cX$. For that we introduce the normalized post measurement state $\omega_{BE}(x)$ conditioned that Alice measures the amplitude $x\in\mathbb R$. Furthermore, we denote by $p(x)$ the probability that Alice measures $x$. Since $\omega_{AE}$ is a Gaussian state, one can show that $\omega_{E}(x)= {\rm U} (v(x))\omega_E(0) {\rm U}(v(x))^\dagger$, where ${\rm U}(v)$ denotes the Weyl operator which corresponds to a phase space translation and $v$ is a continuous function which depends on $\Gamma_{AE}$. Hence, we obtain that $H(E)_{\omega(x)}=H(E)_{\omega(0)}$ for all $x$.


\begin{prop}\label{prop:collective}
Let $\omega_{AB}$ be a two mode squeezed Gaussian state, $\omega_{ABE}$ a Gaussian purification, and $\omega_{BE}(x)$ and $\omega_{BE}^k$ as defined above. Then, it follows that $H(E)_{\omega^k} \geq H(E)_{\omega(0)}$ and, thus, $H(E|X_A)_{\omega} \geq H(E)_{\omega(0)}$.
\end{prop}

\begin{proof}
The proof exploits the concavity of the von Neumann entropy and the fact that the state $\omega_{E}^k$ can be approximated in trace class by a finite convex combination of states $\omega_{E}(x)$. Note that we can write $\omega_{E}^k=1/p_k \int_{I_k} p(x)\omega_{E}(x){\rm dx}$ where the integral converges weakly. As discussed above we also have the relation $\omega_{E}(x)={\rm U}(v(x)) \omega_{E}(0){\rm U}(v(x))^{\dagger}$. Since ${\rm U}(v)$ is strongly continuous in $v$, we have $x\mapsto \omega_{BE}(x)$ and, thus, $x\mapsto \omega_{E}(x)$ are trace class continuous. Hence, we know that the Lebesgue integral $\int_{I_k}p(x)\omega_{E}(x) {\rm dx}$ converges even in trace norm, and furthermore, it is equal to the Riemann integral. So we can approximate $\omega^k_E$ in trace norm via step functions
\begin{equation*}
\rho_E^l=\frac{1}{p_k}\sum_{j=1}^{N_l}p(x^l_j)|J^l_j|\omega_E(x^l_j)\ ,
\end{equation*}
where it holds for all $l$ that $I_k=\bigcup_j J^l_j$, the $x_j^l\in J_j^l$ are chosen such that $\sum_{j=1}^{N_l}p(x^l_j)|J^l_j|=p_k$, and $\sup_j|J_j^l|\rightarrow 0$ for $l\rightarrow \infty$. Furthermore, as $\omega_E(x)$ is a Gaussian state, we have that for $H = Q_E^2+ P_E^2$ the expectation value $\tr\left[\omega_E(x) H\right]$ is bounded and continuous in $x$, so $\tr\left[\rho_E^l H\right]\rightarrow\tr\left[\omega_E^kH\right]$ for $l\rightarrow \infty$~\footnote{We also use that $x<\infty$ for $x\in I_k$ since $I_k\subset \mathbb R$ for all $k$.}. Using that the von Neumann entropy is continuous for sequences of states with finite energy~\cite{Wehrl78}, we find that $H(E)_{\omega^k}=\lim_{l\rightarrow \infty} H(\rho_E^l)$, and thus,
\begin{align*}
H(\omega_E^k) &=\lim_{l\rightarrow \infty} H(\rho_E^l) \\
&\geq \lim_{l\rightarrow \infty}  \frac{1}{p_k}\sum_{j=1}^{N_l}p(x^l_j)|J^l_j|H(\omega_E(x^l_j))= H(\omega_E(0))\ .
\end{align*}
The inequality is due to the concavity of the von Neumann entropy~\cite{Wehrl78} and the last equality holds because $H(\omega_E(x))$ is independent of $x$.
\end{proof}

Using this proposition we finally get
\begin{align*}
H(X_A|E)_{\omega} \geq H(E)_{\omega(0)} +H(X_A)_{\omega} - H(AB)_{\omega}\ ,
\end{align*}
where the right-hand side can be calculated since $\omega_E(0)$ and $\omega_{AB}$ are Gaussian states (see~\cite{serafini04}). Note that the only dependence on the interval length $\delta$ in this formula is due to $H(X_A)_{\omega}$.
\end{appendix}

\bibliographystyle{apsrev4-1}
%

\end{document}